\documentclass[12pt,a4paper]{article}

\usepackage[
  bookmarks=true,
  bookmarksnumbered=true,
  bookmarksopen=true,
  bookmarksdepth=2,
  pdfborder={0 0 0},
  breaklinks=true,
  colorlinks=true,
  linkcolor=black,
  citecolor=black,
  filecolor=black,
  urlcolor=black,
]{hyperref}

\usepackage[margin=1in]{geometry}

\usepackage{amsmath}
\usepackage{amsthm}
\usepackage{amssymb}
\usepackage{mathtools}

\usepackage{enumitem}
\usepackage[capitalize]{cleveref}

\usepackage{nicefrac}

\usepackage{fnbreak}
\allowdisplaybreaks

\theoremstyle{plain}
\newtheorem{theorem}{Theorem}
\newtheorem{lemma}[theorem]{Lemma}
\newtheorem{example}{Example}

\Crefname{sublemma}{Sublemma}{Sublemmas}

\newtheorem{corollary}[theorem]{Corollary}
\newtheorem{open problem}[theorem]{Open Problem}

\theoremstyle{definition}
\newtheorem{definition}{Definition}[section]

\newlist{parts}{enumerate}{1}
\crefname{partsi}{Part}{Parts}
\setlist[parts,1]{label=\alph*.,ref=\alph*}

\title{Serial Monopoly on Blockchains}

\author{
Noam Nisan\thanks{Starkware and School of Computer Science \& Engineering, The Hebrew University of Jerusalem \emph{E-mail}: \href{mailto:noam.nisan@gmail.com}{noam.nisan@gmail.com}.}
}

\begin{document}

\maketitle

\begin{abstract}
We study the following problem that is motivated by Blockchains where ``miners'' 
are serially given the monopoly for assembling transactions into the next block.  
Our model has a single good that is sold repeatedly every day where 
new demand for the good arrives every day.  The novel element in our model is that all 
unsatisfied demand from one day remains in the system and is added to the new demand of the next day.
Every day there is a new monopolist that gets to sell a fixed supply $s$ of the good and naturally 
chooses to do so at the monopolist's price for the combined demand.
What will the dynamics of the prices chosen by the sequence of monopolists be?  What level of efficiency will be obtained in the long term? 
 
We start with a non-strategic analysis of users' behavior and our
main result shows that prices keep fluctuating wildly and this is an endogenous property of the model and happens even when demand is stable with nothing stochastic in the model.  These price fluctuations underscore the necessity of an analysis under strategic behavior of the users,
which we show results in the prices being stable at the market equilibrium price.
\end{abstract}

\section{Motivation: Transaction Fees on Blockchains}

Blockchain systems like Bitcoin \cite{bitcoin} or Ethereum \cite{eth} sell ``slots'' on the blockchain to users who 
wish to put their transactions on it.  Every period a ``leader'' (miner, validator, sequencer) is chosen
to produce the next ``block'' in the blockchain, where the choice of the leader is done using 
some mechanism that need not concern us here such as proof-of-work or proof-of-stake.  
The size of each block is limited by the protocol in some way (e.g., bytes for Bitcoin or ``gas'' for Ethereum), 
and the 
leader gets to choose which transactions will fill the block up to that limit.  
The blockchain's transaction fee mechanism specifies how much the users of chosen transactions pay and how much the leader receives (in addition to a fixed ``block reward'') 
and needs to take into account that both the users and the leaders are strategic.  

The mechanism used by the Bitcoin blockchain is simple {\em ``pay your bid''}: users place bids for their 
transactions and the leader (miner) gets to choose an arbitrary subset of transactions and 
charges each of them exactly what was bid for it.  Clearly a
strategic leader will accept the highest bidding transactions (normalized to their ``size'') that
fit within the block size limitations.
A strategic user will obviously shade his bid by an amount that is not easy to calculate well.

The Ethereum blockchain
has a mechanism, known as EIP-1559 \cite{1559}, that aims to be more straightforward for users to bid.
The mechanism's most significant feature is that it uses {\em ``dynamic posted prices''} where
the (``base-gas'') price for the next block is determined by the protocol as a function of previous blocks.  
Each user
bids a maximum price and only transactions that bid at least the 
block's price will be included in the block, and they all pay the block's fixed gas price (rather than their 
bid).  A significant additional feature of EIP-1559 is that all fees are ``burnt'' rather than going 
to the leader who only gets a small additional first-price-like ``tip''.  Burning the fees is required to ensure that leaders are not motivated to collude with
users.  Conceptually, since the (base-gas) price of a block is deterministically
determined by the protocol according to the history, neither the leader nor the users have
any advantage in manipulation.  The formula that determines the block prices increases them when there has recently been more demand than supply and decreases them in the opposite case
thus managing to balance the satisfied demand with the
average supply of slots in a block.  The exact incentives in this mechanism and related ones are formally
defined and studied in \cite{R21} and further in \cite{CS21}, an analysis of the block price dynamics appears in \cite{LRMP22} and 
the general class of ``dynamic posted prices" mechanisms is studied in \cite{FMPM21}.  We will
not delve deeper into the details of this mechanism as for our purposes
the simple conceptual description above suffices.

The fee-burning part of this mechanism may be viewed as undesirable as 
it reduces the Ethereum token supply which may or may not
be desired from other points of view.  
A third mechanism -- that does not require fee burning -- suggested in \cite{BE19} and in \cite{LSZ22} 
is to use
{\em monopolist pricing}: each leader is allowed to choose an arbitrary price for his block
and can then collect all transactions that are willing to pay this price (in \cite{BE19} this is called generalized second price).  
The rational leader will
certainly choose the monopolist price that maximizes the product of the resulting block size and the price.
Intuitively, as transactions are expected to be small relative to the total
block size one may expect users to be ``price-takers'' and thus
not to have any significant incentive to shade their bids.  This mechanism was analyzed in 
\cite{LSZ22, Y18, BE19}, but again for our purposes this simple intuition suffices.

A major difference between the monopolist pricing mechanism and the two previous ones is in
what they optimize for.  The first two should reach (close to) the market equilibrium and
thus optimize ``social welfare'' -- the total value of accepted transactions
subject to the blockchain capacity limitations\footnote{While the exact analysis may depend on the model, 
intuitively both ``pay your bid'' and ``EIP-1559'' 
should reach the market equilibrium.}.  The monopolist
pricing optimizes the leader's revenue and may lead to unbounded losses of welfare.  While
not optimizing social welfare is certainly a weakness of this mechanism, as \cite{LSZ22} argues, optimizing
revenue may be an advantage for the security of the blockchain.  
In particular they note that mechanisms
that reach market equilibrium have the problematic property that if the platform's capacity
suffices to handle {\em all} demand, then the prices would go down to $0$ which may endanger the security of the blockchain.

All the discussion so far looked at a single block in isolation: it looked at the single leader of the 
block and the set of users for that block and assumed that they all
were myopic i.e that their strategic considerations were only about the given block.  
This is the case both for the intuitive explanations above and for the
formal analysis in the papers cited.  While this assumption may be a good modeling choice for the leaders since 
in large systems we expect a single miner to only be chosen to be leader infrequently,
it is not at all realistic for the users since a transaction that is not accepted to one block remains in the 
``mem-pool'' and can be accepted into one of the next few blocks, potentially 
within less than a minute.  Indeed \cite{LSZ22} left the analysis of ``patient users'', as an open problem.  
Even ignoring the strategic behavior of users, just the fact that the unsatisfied transactions
from one block remain as demand for future blocks, clearly leads to a nontrivial dynamic in the sequence of blocks.\footnote{This is 
somewhat in the spirit of, e.g., \cite{GT20,GT21} in a different context, and considered by
\cite{HLM21} in the blockchain context for the case of pay your bid auctions.}  
Specifically, if a monopolist in some block has charged a high price, leaving much unsatisfied demand, then the next block will
get this pent-up demand and thus see a heaver total demand at lower values which
intuitively may cause the next 
block's price to be lower.  
Figure \ref{fig:eth} shows a simulation of this dynamic on a data from a typical ``uneventful'' sequence of
50 Ethereum blocks.  The serial monopoly dynamics extracted $12\%$ more revenue on this sequence but lost $6.5\%$ 
of total welfare (total sum of values of transactions, counted according to their bids).  The heart of this paper is trying to 
analyze this behavior, understanding the fluctuations in prices, the welfare loss, and the implied strategic
considerations of the users.

\begin{figure}
\centering
\includegraphics[scale=0.85]{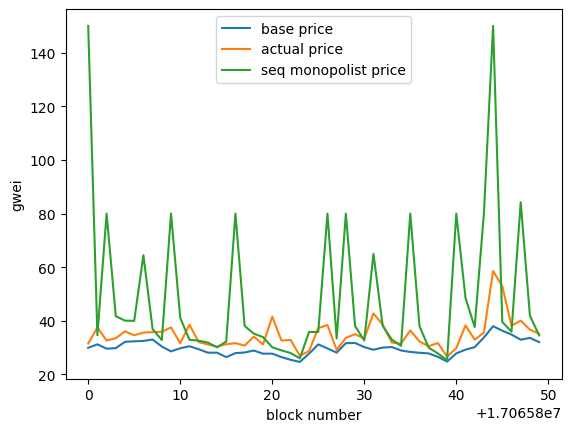}
  \caption{\small Simulation of serial monopoly on data from a sequence of 50 
	Ethereum blocks (10 minutes), comparing 
	the series of monopolist prices to the true block gas prices with average tips (actual price) 
	and without tips (EIP-1559 base price).}
  \label{fig:eth}
\end{figure}

In our model we have a series of monopolists where each of them is
faced with all the pent-up demand from previous blocks as well as some new flow of demand and then
gets to chose his monopolist price.  
We start by analyzing the dynamics assuming that bidders are non-strategic or,
equivalently, act myopically as previously studied, even though their valuation is really patient, i.e.
they get the value of their transaction even if it
is scheduled at later blocks.\footnote{Our analysis focuses
on ``fully patient'' bidders whose value for the transaction does not decay with time.  This also models well scenarios
where value decay happens at time scales that are significantly larger than the block times, a situation that seems to fit most
transactions on the blockchain with notable exceptions being MEV and some DeFi.  
One may certainly also consider
intermediate levels of patience where a transaction's value decays with time at some rate, and
such cases would be expected to lie between the fully myopic and full patient extremes.}  

There are several reasons for exploring
such myopic behavior by non-myopic users. First, this is a natural first step before continuing with a strategic analysis.  This is
especially true in this case where previous argumentation as well as intuition may suggest that users are close
to being ``price takers''
and gain little, if at all, from strategizing. Second, this can explain simulations, such as that given in figure \ref{fig:eth},
that are run on existing data.  Finally, this analysis will turn out to have implications for the strategic analysis that we will do later.

We will formally describe
our model and results in the next subsection but, for now, let us state their intuitive implications for
blockchains.  Our main, and surprising, result is that monopolist pricing dynamics leads to never-ending 
price fluctuations and this happens endogenously without any stochastic element in the model and when
the exogenous conditions are completely stable.  In this respect serial monopoly completely 
sacrifices one of the main desiderata for a fee mechanism, that of price stability.  Even worse,
once prices fluctuate, users are motivated to shade their bids and wait for lower prices.  This happens
even when bidders are small and each one of them does not affect prices at all.  Despite being
a price taker within a single block, shading bids is highly beneficial across blocks.

We do find some silver lining here regarding the social welfare achieved.  Recall that a main concern regarding
monopolist pricing is that it sacrifices efficiency since leaders do not fill block completely.  
We do, however, show that the social welfare (total value of accepted transactions) 
achieved by the serial monopoly rule (still with users bidding non-strategically)
is mathematically guaranteed to be at least one half of the optimal social welfare.  As usual when
one can prove a formal guarantee, things are better in most specific cases.  E.g., for transaction
values that are uniform in $[0,1]$ we calculate the loss of social welfare to be only $6.25\%$.  
And, nicely, this happens while gaining on revenue.

As mentioned, this finding of price fluctuations strongly suggests that patient users should bid strategically and
so calls for an analysis under strategic bids by the users. 
Continuing to analyze ``fully-patient'' users, but now acting strategically, we
find that, in equilibrium, users 
shade their bids so that the system is always at the market prices without any price fluctuations. 
However, as this shading requires information about the market conditions, 
optimal bidding may be difficult. 

We are thus back where we started, having reached essentially the same outcome and difficulties as the pay-your-bid mechanism.  I.e., once
strategic patience is taken into account the main motivations of \cite{LSZ22} of near incentive-compatibility and of better revenue are lost.
While one may view this as an overall negative conclusion given the original motivation, the outcome reached by serial monopoly 
should be the same as that reached by the pay-our-bid mechanism, and so it does look like a completely viable alternative for a fee mechanism.  
More refined comparison of the serial monopoly
mechanism to the ``vanilla''
pay-your-bid mechanism may require experimentation and one may speculate that
bidding can be easier for serial monopoly 
since, at equilibrium, my own payment does not depend on my own bid.

\section{Serial Monopoly: Model and Results}

We now formalize and analyze our model in abstract terms of a ``serial monopoly'' that may be of more general interest.

\subsection{The (Non-Strategic) Model}

So here is our model: 
At each time step $t=1,2,...$ some demand for some  
homogeneous good arrives into the market.  The daily demand is 
specified by a fixed demand function $Q(p)$ that specifies the 
demanded quantity at each price level $p$.  For ease
of exposition we will assume that $Q$ is continuous and strictly decreasing.

Every day a new monopolist with a fixed daily supply $s$ is chosen.  This monopolist
sees in front of him the total current demand $D^t()$ which is the sum of the pent-up demand from previous time steps and the new daily demand
and chooses a price level $p^t$ that
maximizes his revenue.   Specifically the price chosen by the monopolist at day $t$ is the price $p^t$ 
that maximizes the revenue $p \cdot min(s,D^t(p))$ 
and the quantity supplied is thus $q_t=D^t(p^t) \le s$.   
The pent-up demand after this amount is supplied is given 
by a demand function $Z^t(p)=D^t(p)-q^t$ for $p \le p_t$ and
$Z^t(p)=0$ for $p \ge p_t$.  The total demand for the next time step is obtained by 
adding the daily demand $Q()$
to this pent-up demand.  So, to summarize, here is the formal dynamics we study:

\begin{itemize}
\item A continuous and strictly decreasing demand function $Q()$ and a fixed supply amount $s$ are given.
\item There is initially no pent-up demand: $Z^0(p)=0$ for all $p$.
\item For every time step $t=1...$:
\begin{itemize}
\item The day $t$ demand function is given by $D^t(p)=Z^{t-1}(p)+Q(p)$.
\item The day $t$ monopolist price and quantity are given by $p^t = argmax_p (p \cdot min(s,D^t(p)))$ and $q_t=D^t(p^t)$.
\item The pent-up demand function after day $t$ is given by $Z^t(p)=D^t(p)-q^t$ for $p \le p^t$ and $Z^t(p)=0$ for $p \ge p^t$.
\end{itemize}
\end{itemize}

This model simplifies matters as much as possible, in particular
assuming (1) a fixed flow of demand, i.e. that the same demand distribution arrives every day, 
(2) that each monopolist has the same fixed supply amount, (3) that monopolists never repeat, i.e. are completely
myopic and thus naturally behave as a simple monopolists, (4) the demand is ``infinitely patient'' so values do not decay with time and (5) 
that the daily demand is fixed, not chosen strategically, and fully known by the monopolist.  Also note that the model is completely deterministic.

We would like to analyze what happens when this dynamics reaches an equilibrium: 
which prices and quantities would be reached and what
is the resulting efficiency?  We are in for an unpleasant surprise: the dynamics do not 
reach an equilibrium, but instead get some complex 
non-cyclic pattern of ever changing prices.  Here is a typical example.

\subsection{Example}

\begin{figure}
\centering
\includegraphics[scale=0.85]{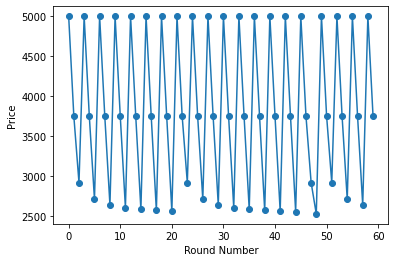}
  \caption{\small Serial monopoly prices for demand that is uniform on \{1,2,...,10000\}}
  \label{fig:unif-sim}
\end{figure}

Assume that the daily demand is generated by user valuations that are uniform in $[0,1]$, i.e. the daily demand
function is given by $Q(p)=1-p$ (for $0 \le p \le 1$) and let us assume that the fixed daily supply is $s=1$.  
The market equilibrium is at $price=0$ and $quantity=1$, which give a total social
welfare of $1/2$, and revenue of $0$.
The monopolist aims to maximize $p \cdot (1-p)$ which happens for for $price=0.5$ and $quantity=0.5$ with sub-optimal social welfare of $3/8$ and monopolist's revenue of $1/4$.  
So let us follow the serial monopolists step by step.

\begin{enumerate}
\item
Day 1: There is no pent-up demand  so the first monopolist will indeed choose the monopoly 
price of $p^1=0.5$ with quantity $q^1=0.5$ obtaining revenue of $1/4$.  
The pent-up demand after the first step is 
given by $1/2-p$ for $p \le 1/2$ (and 0 for $p \ge 1/2$).

\item
Day 2: The 
total demand at this stage is $3/2-2p$ for $p \le 1/2$ and is $1-p$ for $p \ge 1/2$.  
The monopolist now can do better than
choosing the original monopoly price by 
choosing $p^2=3/8=0.375$ with quantity $q^2=3/4$ and revenue of $9/32>1/4$.  
The pent-up demand now is $3/4-2p$ for $p \le 3/8$ (and 0 for $p \ge 3/8$). 

\item Day 3: 
the total demand is now given by $7/4-3p$ for $p \le 3/8$ (and the usual $1-p$ for $p \ge 3/8$).  
The optimal price turns out to be  
$p^3=7/24=0.29...$ with quantity $q^3=7/8$.   Pent-up demand is $7/8-3p$ for $p \le 7/24$.

\item
Day 4: now we have a total demand of $15/8-4p$ for $p \le 7/24$ (and the usual $1-p$ for $p \ge 7/24$). Surprisingly, the pent-up demand
does not help the monopolist: the demand at price $7/32$ is already 1 so the highest revenue
obtainable in the range $p \le 7/24$ is at this price which would give revenue of only $7/32<1/4$ so
the optimal revenue is obtained at the original monopolist price of $p^4=0.5$.

\item
A simulation of the first 60 steps with discretized values appears in figure 1 where we see that the prices keep 
irregularly oscillating. We wish to emphasize that the observed irregularity is {\em not} an 
artifact of the discretization or of the simulation.
\end{enumerate}

\subsection{Results}

We first show that even though the prices keep oscillating irregularly, we can provide sufficient analysis of the price dynamics, in particular
showing that there exists a price
$p^{ser}$ such that all demand above it is satisfied within a bounded time
and all demand below it is never supplied at all.  We show that this price is given by the formula $p^{ser} = p^{mon} \cdot q^{mon} / s$,
where $p^{mon}$ is the monopoly price of $Q()$ and $q^{mon}$ the monopoly quantity.
Thus, despite the lack of any convergence to equilibrium, $p^{ser}$ 
can be viewed as the one ``reached'' by the dynamics in this sense of which demand gets supplied.  The following theorem applies to any
(strictly decreasing and continuous) demand function $Q()$ and supply amount $s$ for which {\em the monopolist revenue is strictly
higher than the market equilibrium revenue.}

\begin{theorem} \label{thm-dyn}
The dynamics of the the daily prices $p^t$ behave as follows:
\begin{enumerate}
\item They are sandwiched between the price $p^{ser} = p^{mon} \cdot q^{mon} / s$ and the monopolist price
$p^{mon}$, i.e., for all $t$ we have that $p^{ser} \le p^t \le p^{mon}$.  In particular, no demand at prices lower than $p^{ser}$ is ever supplied.
Furthermore, these bounds are tight even in the limit and
$p^{ser} = \liminf_{t \rightarrow \infty} p^t$ and $p^{mon}=\limsup_{t \rightarrow \infty} p^t$.  
\item For every price $p>p^{ser}$
there exists a constant $\Delta_p$ such that in every consecutive $\Delta_p$ steps we have at least some 
$p^t \le p$ and thus all demand at a
price above $p^{ser}$ is eventually supplied and furthermore this happens within a bounded 
time lag (where the bound $\Delta_p$ depends on the price). 
\end{enumerate}
\end{theorem}

\begin{figure}
\centering
\includegraphics[scale=0.50]{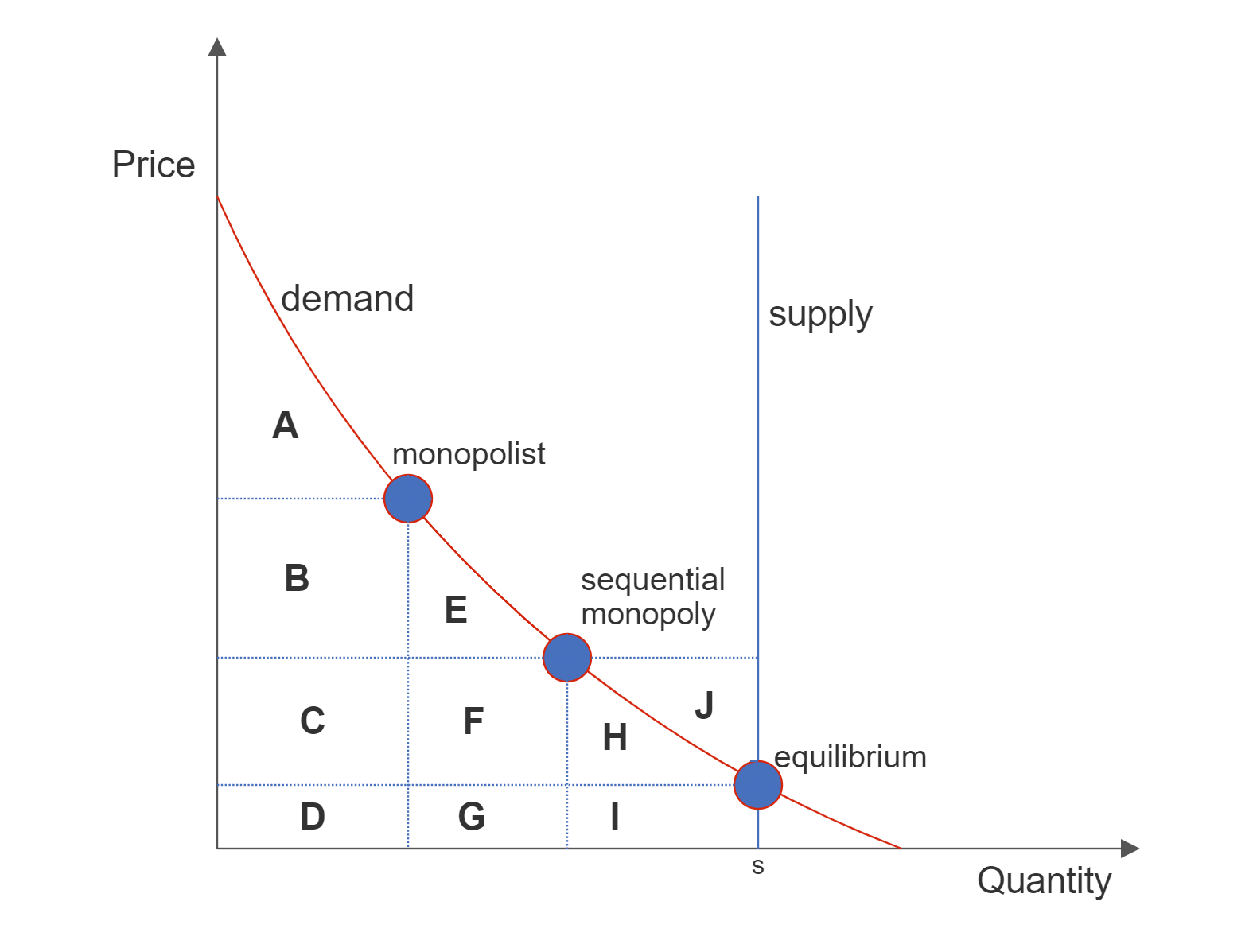}
  \caption{\small The Supply and demand curves with the market equilibrium, monopoly, and serial
monopoly (price,quantity) points.  Different areas that depict parts of the 
social welfare at these points are given labels.}
  \label{fig:sup-dem}
\end{figure}

We then analyze the social welfare achieved by this dynamic.  The ``social welfare'' is the sum (integral) of the 
values of the users whose demand is supplied.  It is well known that monopolist pricing may cause an unbounded loss of social welfare.
Surprisingly, we show that the social welfare achieved by serial monopolists approximates the optimum well!
Since our analysis above shows that all demand above price $p^{ser}$ is supplied, the average long term social welfare is 
depicted by the total area of regions A+B+C+D+E+F+G in figure \ref{fig:sup-dem}.  This welfare
is obviously bounded by the optimal possible social welfare (depicted as the total areas of regions A+B+C+D+E+F+G+H+I in figure \ref{fig:sup-dem}), but turns out
to be not far from it.

\begin{theorem} \label{thm-sw}
The social welfare obtained by by the serial monopoly
is at least one half of the  
optimal social welfare.  This bound is 
tight as there exists some daily demand $Q()$ where the ratio is exactly $1/2$.
\end{theorem}

This is a worst case result and for specific distributions the bounds are better.  In particular, for the uniform distribution the ratio is $15/16$ while
for the ``equal revenue'' distribution, the classic example with a large gap between monopolist social welfare and optimal social welfare,
serial monopoly turns out to yield optimal social welfare.

\subsection{Strategic Analysis}

All the analysis so far was for a fixed demand function and only analyzed the rational actions
of the monopolists.  If we imagine the demand coming
from a continuum of utility-maximizing bidders (as in our motivating application) then
this translates to these bidders acting in a myopic, non-strategic, way as in \cite{LSZ22}.  However the fact that, as we show, 
intra-block prices fluctuate wildly puts this modeling assumption
in question as even price-taking bidders will likely be ``patient'', preferring to wait for cheaper blocks. 
Intuitively, patient users are motivated to
strategically shade their bids down to (almost) $p^{ser}$ which is the lowest price that they can get in any block.
Once they all do so, the distribution that the monopolist sees in front of him is no longer $Q()$ but rather
a strategically declared lower distribution and as he can only optimize his revenue relative to this distribution,
intuitively leading to a lower serial price $p'^{ser}$, further reducing the bids of the users, etc. 

This leads us to an equilibrium analysis of this system where both users and leaders are strategic. 
We stick with our simple non-stochastic deterministic model with price-taking bidders and have the 
daily demand distribution $Q()$ and supply 
amount $s$
being common knowledge.  We consider a game between
multiple users and leaders: for each time step $t$ we have a (new) continuum of infinitely-small users and a single (new) leader.  
Each of our users has a true value $v$, where $v$ is chosen for each user in the continuum of users every time step 
according to the demand distribution $Q$.
In each step $t$, each user with true value $v$ declares a a value $\tilde{v}$ according to some strategic manipulation function 
that can be a function of his particulars: his value and the time, $\tilde{v}=m_t(v)$.
The {\em declared demand} at time $t$, $\tilde{Q}^t$, is generated by the distribution of $\tilde{v}$ at time $t$.  
The leader at each time $t$ may choose a price $p^t$ and 
a {\em dominant strategy} is to choose the monopolist price for the {\em declared} current demand.\footnote{Since the leader may only collect a fee
that is bounded by the {\em declared} value, the leader's knowledge of the true demand and even true values does not allow him to do any better.}

We will focus on equilibria where the users' manipulation function is time-invariant,
i.e. $m_t(v)=m(v)$ for some fixed manipulation function $m$, and where leaders all use their dominant (only non-weakly-dominated) strategy
of monopoly pricing.  This would be an equilibrium if for every bidder that is born in time $t$ with value $v$, $m(v)$ is indeed a best reply
to all the other users bidding according to $m$ and all the leaders choosing the monopoly price.

It turns out that such an equilibrium of serial monopoly just goes back to the market equilibrium of the
true demand, hence loosing the price fluctuations, optimizing social welfare, and giving up on any revenue
optimization.

\begin{theorem} \label{thm-eq}
There exists an equilibrium in un-dominated strategies with time-invariant manipulation where $p^t = p^{eq}$ and $q^t = s$ for every $t$, 
where $p^{eq}$ denotes the market equilibrium price, i.e. $Q(p^{eq})=s$.  In this 
equilibrium users bid $\tilde{v}=m(v)=min(v, p^{eq})$ and leaders charge monopoly prices.

Furthermore, in every 
equilibrium in un-dominated strategies with time-invariant manipulation we have that $p^t = p^{eq}$ and $q^t = s$ for every $t$.  
\end{theorem}

As in the non-strategic case, this analysis continued to consider bidders that ``fully patient'', 
i.e. assumed that user values do not decay with time.  
The next three sections are devoted, respectively, to the proofs of the three theorems, where most of the technical work is 
in the analysis of the dynamics in section \ref{sec-dyn}.

\section{Analysis of the Dynamics}\label{sec-dyn}

This section gradually analyzes the (non-strategic) price dynamics in our model.  All proofs of lemmas appear in the appendix.

\subsection{The Evolution of Pent-up Demand}

We will analyze the pent-up demand within intervals of prices, i.e., for $p'>p$, we are interested in $D^t(p)-D^t(p')$.  
The pent-up demand in the range $[p,p']$ evolves as follows: First, every time step a new amount of $Q(p)-Q(p')$ 
is added to the pent-up demand.  
Whenever some $p^t$ is smaller than $p$ then all of this demand is supplied and the pent-up demand for time $t+1$ 
is zero.  When $p^t$ is larger than $p'$
all of the pent-up demand just remains for the next step and whenever $p < p^t < p'$ then some of this demand is supplied 
and some is remains for the next step. Let us write this down formally.

When $p^{t-1} < p < p'$ we will have no pent-up 
demand $Z^{t-1}(p)=Z^{t-1}(p')=0$ and thus $D^t(p)-D^t(p')=Q(p)-Q(p')$.  When
$p^{t-1} > p' >p$ we have that $Z^{t-1}(p)=D^{t-1}(p)-q^{t-1}$ and $Z^{t-1}(p')=D^{t-1}(p')-q^{t-1}$ 
so $Z^{t-1}(p)-Z^{t-1}(p') = D^{t-1}(p)-D^{t-1}(p')$ and $D^t(p)-D^t(p')=(D^{t-1}(p)-D^{t-1}(p'))+(Q(p)-Q(p'))$.
Finally, when $p \le p^{t-1} \le p'$ we will have $Z^{t-1}(p') = 0$ while $Z^{t-1}(p)=D^{t-1}(p)-q^{t-1}$ and, 
since $D^{t-1}(p') \le q^{t-1} = D^{t-1}(p^{t-1})$, we have $Z^{t-1}(p)-Z^{t-1}(p') \le D^{t-1}(p)-D^{t-1}(p')$.
Note that the last inequality holds in all three cases.

Summing this up over a prefix of times $\{1,2,...t\}$, or over a range of times $\{T+1, T+2, ... t\}$ we get the following lemma.

\begin{lemma} \label{lem-a}
For every $p \le p'$ we have that:
\begin{enumerate}
\item For every $t$ we have that $D^{t}(p)-D^{t}(p') \le t \cdot (Q(p) - Q(p'))$.  
\item For all $T$ and $t>T$ we have that:
$D^{t}(p)-D^{t}(p') \le (t-T) \cdot (Q(p) - Q(p')) + (Z^{T}(p)-Z^{T}(p'))$.  
\item For all $T$ and $t>T$, if for all $t'$ such that $T < t' < t$ we also have that $p^{t'} \ge p' >p$ then in fact we have equality
$D^{t}(p)-D^{t}(p') = (t-T) \cdot (Q(p) - Q(p')) + (Z^{T}(p)-Z^{T}(p'))$. 
\item For all $T$ such that $p^T \le p < p'$ (or $T=0$) and all $t>T$ we have that 
$D^{t}(p)-D^{t}(p') \le (t-T) \cdot (Q(p) - Q(p'))$. 
\item For all $T$ such that $p^T \le p < p'$ (or $T=0$), if for all $t'$ such that $T < t' < t$ we also have that $p^{t'} \ge p' >p$,  then 
$D^{t}(p)-D^{t}(p') = (t-T) \cdot (Q(p) - Q(p'))$. 
\end{enumerate}
\end{lemma}

\subsection{The Serial Monopoly Price and Quantity} \label{sec-seq}

Let us start by looking at the two basic (price,quantity) points of the daily market with demand $Q()$ and supply $s$: the market equilibrium and the monopolist pricing.

{\bf The market equilibrium} point is when supply equals demand. i.e. at the price $p^{eq} = P(s)=Q^{-1}(s)$ 
(where $P()$ is the inverse function of $Q()$, at which point $q^{eq}=Q(p^{eq})=s$, and the social welfare 
$SW^{eq}=\int^{s}_{0} P(q)dq$ is maximized.  In figure 
\ref{fig:sup-dem} the social welfare at equilibrium is visualized the sum of the areas of regions A+B+C+D+E+F+G+H+I.  
The revenue at that point is $REV^{eq}= p^{eq} \cdot s$ which is visualized as the sum of regions D+G+I.

{\bf The monopolist} chooses a price $p^{mon}$ that maximizes $REV^{mon}=p^{mon} \cdot q^{mon}$, where $q^{mon}=Q(p^{mon})$ which on the graph is given by the sum of regions B+C+D.  
The social welfare at this point is $SW^{mon}=\int^{q^{mon}}_{0} P(q)dq$ which in figure 
ref{fig:sup-dem} is given by the sum of regions A+B+C+D.  
As we assume that the monopolists revenue is strictly larger than the market equilibrium revenue $REV^{mon} > REV^{eq}$,
we must have $p^{mon} > p^{eq}$, $q^{mon} < q^{eq}$,  and $SW^{mon} < SW^{eq}$.  
The gap between $REV^{mon}$ and $REV^{eq}$ can be unbounded as the latter may even be $0$ (if the demand is bounded 
by $s$, i.e. when $P(s)=0$).  The gap between $SW^{mon}$ and $SW^{eq}$ is also known to be potentially unbounded.  Specifically, setting $H=P(0)/P(s)$, the gap can be large as $\ln(H)$, but no more.

We now define the ``serial monopoly'' point.  While there is not going to be any convergence of the prices $p^t$, we will still be able to focus on a meaningful definition of $(p^{ser},q^{ser})$ that captures useful information about
the prices and quantities in the long term of 
the dynamic.  We will define $p^{ser}$ as the price at which selling all the supply would give the monopoly revenue.  

\begin{definition}
The serial monopoly price and quantity are defined as $p^{ser}= p^{mon} \cdot q^{mon} / s$;  $q^{ser}=Q(p^{ser})$.
\end{definition}

Since $q^{mon}<s$ we have that $p^{ser} < p^{mon}$ and $q^{ser} > q^{mon}$.
We also have $p^{ser} > p^{eq}$ and $q^{ser} < s$ since otherwise we would 
have $p^{eq} \cdot s = p^{ser} \cdot s = p^{mon} \cdot q^{mon}$,
contradicting our assumption that the monopolist's revenue is strictly larger than the equilibrium revenue.

\vspace{0.1in}	
\noindent
{\bf Comment:} In the more general case where $Q()$ is not strictly decreasing, the definition of $p^{ser}$ 
should be corrected to be
the largest value $p^{ser}$
such that $Q(p^{ser})=Q(p^{mon} \cdot q^{mon} / s)$; we will continue to assume that $Q$ is 
strictly decreasing so for us this 
complication is superfluous.

\subsection{Basics of Price Dynamics}

The first easy lemma states that $p^{ser}$ is a clear lower bound for any $p^t$.  
This also provides the intuition for the particular choice of $p^{ser}$.

\begin{lemma} \label{A1}
For every $t$ we have that $p^t \ge p^{ser}$.
\end{lemma}

Using lemma \ref{lem-a} (5) with $T=0$ we see that all demand below $p^{ser}$ remains as pent-up demand.

\begin{corollary} \label{cor-b}
For every $p<p^{ser}$ and for all $t$, $D^{t}(p)-D^{t}(p^{ser}) = t \cdot (Q(p)-Q(p^{ser}))$.
\end{corollary}

We next analyze the price movement, where the easy, but perhaps surprising observations is that prices always decrease, unless they ``jump up'' to the monopoly price (and then start decreasing again).

\begin{lemma} \label{lem-b1}
For every $t$ either $p^t=p^{mon}$ or $p^t < p^{t-1}$.
\end{lemma}

This in particular implies that prices never go above $p^{mon}$.

\begin{corollary} \label{A2}
For all $t$: $p^t \le p^{mon}$.
\end{corollary}

The next lemma shows that whenever there is sufficient pent-up demand above $p^{ser}$ then 
the price cannot go up to $p^{mon}$ and furthermore the decrease in the prices is significant - when measured in terms of of the density of the demand.

\begin{lemma} \label{lem-b2}
Assume that for some $p > p^{ser}$ we have that $D^t(p) \ge s$ then $p^t < p^{t-1}$ and
$(Q(p^{t})-Q(p^{t-1})) \ge (t-1)^{-1} \cdot s \cdot (p-p^{seq}) / p^{mon}$.  
Furthermore, if for some $T<t$ we had $p^T \le p^t$ then actually 
$(Q(p^{t})-Q(p^{t-1})) \ge (t-1-T)^{-1} \cdot s \cdot (p-p^{ser}) / p^{mon}$.
\end{lemma}

\subsection{Sequences of decreasing price steps}

The first simple lemma states that as long as prices remain above some threshold $p'$ then the demand for all lower prices $p$ just keeps being pent-up until it reaches as high a quantity as we desire which in our case is $s$.

\begin{lemma} \label{lem-c}
For every $p<p'$ there exists $\Delta_0$ such that for all $T$ and all $\Delta \ge \Delta_0$ we have that
either (a) there exists $T \le t \le T+\Delta$ with $p^t < p'$ or (b) $D^{T+\Delta}(p) \ge s$.
\end{lemma}

We now reach the key part of the characterization showing that prices indeed approach $p^{ser}$ (decreasing from above).

\begin{lemma} \label{lem-main}
For every $p^* > p^{ser}$ there exists $\Delta$ such that for every $T$ there exists some $T < t \le T+\Delta$ with $p^t \le p^*$.  
\end{lemma}

Let us say a word of intuition.  Lemma \ref{lem-b2} shows that when we are in a sequence of 
decreasing prices $p^t$, each time the rate of decrease -- when measured in terms of $Q()$ -- is proportional to $1/t$ (or even
$1/(t-T)$ where $T$ is the last time the price was below a threshold that we are aiming for).  
Thus, as the series $1/t$ diverges, we cannot have an infinite sequence of decreasing prices until we go below 
our desired threshold.  The formal proof appears in the appendix.

\subsection{Prices need to jump up}

We start by stating the obvious fact that demand that was not supplied remains as pent-up demand:

\begin{lemma} \label{lem-q}
For any time $T$  
we have that  $\sum_{t=1}^{T-1} q^t \ge T \cdot Q(p) - D^T(p)$.  If $p \le \min_{t < T} p^t$ then we have equality.
\end{lemma}

In particular, $\sum_{t=1}^{T-1} q^t = T \cdot Q(p^{ser}) - D^T(p^{ser})$.
We are now ready to show that prices must jump up to $p^{mon}$ infinitely often.

\begin{lemma} \label{lem-up}
There exists infinitely many $t$ such that $p^t=p^{mon}$.
\end{lemma}

\subsection{Putting it all Together}

Let us now see how we have proved theorem \ref{thm-dyn}.

\begin{proof} ({\bf of Theorem \ref{thm-dyn}})
Lemma \ref{A1} and \ref{A2} provide the upper and lower bounds on $p^t$ while lemmas \ref{lem-main} and \ref{lem-up} prove that these are tight as $t \rightarrow \infty$.

Lemma \ref{lem-main} provides the bound on the number of steps during which $p^t$ can lie above $p$.  Clearly once $p^t \le p$ the demand 
at this price is completely supplied, $Z^t(p)=0$.
\end{proof}

\section{Welfare Analysis}

As we have so far been able to prove the demand that is supplied by the serial monopoly dynamics is exactly that with $p > p^{ser}$, we get that the social welfare achieved is given by $\int_{0}^{q^{ser}} P(q) dq$, depicted
in figure \ref{fig:sup-dem} by the sum of regions A+B+C+D+E+F+G.  Using the way we defined $p^{ser}$ we now get an approximate welfare result.

\begin{lemma} \label{lem-sw}
For every demand function $Q$ and supply level $s$ we have that the welfare obtained 
by serial monopoly is at least half of the social welfare at equilibrium.
\end{lemma}

\begin{proof}
In figure \ref{fig:sup-dem}, the social welfare at equilibrium is given by the areas of regions A+B+C+D+E+F+G+H+I, while the social welfare at the serial monopoly point is given by regions A+B+C+D+E+F+G.  
It thus suffices to show that the area of H+I is bounded from above by A+B+C+D+E+F+G.  
In fact it is even true that C+D+F+G+H+I+J is bounded by B+C+D.  That is true since the area of the former is $p^{ser} \cdot s$ while the latter
is $p^{mon} \cdot q^{mon}$ which are equal.
\end{proof}

While the proof implies various stronger bounds such as $SW^{ser} \ge max(SW^{mon},SW^{eq}-REV^{mon})$, the 
bound is tight as can be seen from the following example: 

\begin{example} \label{ex-sw}
Assume that the demand is for $M-1$ units at price 1 plus an additional unit at price $M+1$ for some 
large $M$ and assume that the supply is exactly $s=M$.  (This example is discrete giving
a non-continuous and only weakly decreasing demand function, but it is easy to add to it $\epsilon$ mass making it 
continuous and strictly
decreasing as in our analysis.) The equilibrium point is at $price=1$
and $quantity=M$ where the social welfare is $2M$.  
The monopolist would chose $p^{mon}=M+1$ and $q^{mon}=1$ obtaining welfare of only $M+1$.  
The serial monopoly price would be chosen as $p^{ser}=(M+1)/M > 1$ and thus $q^{ser}=1$ 
achieving social
welfare of $M+1$ rather than the possible $2M$. 
\end{example}

So we now have essentially proved theorem \ref{thm-sw}:

\begin{proof} ({\bf of Theorem \ref{thm-sw}})
The combination of example \ref{ex-sw} and lemma \ref{lem-sw} is exactly the statement of the Theorem.  
\end{proof}

In ``typical'' scenarios the loss is significantly lower as is demonstrated by the following examples.

\begin{example} \label{ex-uni}
Consider a demand that comes from the uniform distribution on $[0,1]$, i.e. $Q(p) = 1-p$ for $0 \le p \le 1$ 
and supply of $s=1$.  The inverse function is given by $P(q)=1-q$, the equilibrium price is $p^{eq}=0$ with $q^{eq}=1$ giving total welfare of 
$\int^{1}_{0} (1-q)dq =1/2$. The monopoly price is $p^{mon}=1/2$ with $q^{mon}=1/2$ 
obtaining social welfare of only $\int^{1/2}_{0} (1-q)dq = 3/8$.  
The serial monopoly price would be $p^{ser}=1/4$ with $q^{ser}=3/4$, so the total welfare would be 
$\int^{3/4}_{0} (1-q)dq = 15/32$, which is only $6.25\%$ lower than the optimal $1/2$.
\end{example}

\begin{example} \label{ex-er}
Consider a demand that comes from an ``equal revenue'' distribution: 
$Q(p)=1/p$ for $p \in [1,H]$ (with $Q(p)=1$ for $p \le 1$ and $Q(p)=0$ for $p > H$) and supply $s=1$. 
The inverse function is given by $P(q)=1/q$ for $q \in [1/H,1]$ (and $P(q)=H$ for $q<1/H$).
Equilibrium price is $p^{eq}=1$ where $q^{eq}=1$ at which the social welfare is
$\int^1_0 P(q) dq = \int^{1/H}_0 H \cdot dq + \int^{1}_{1/H} q^{-1} dq = 1+ \ln H$.  
The monopoly price is $p^{mon}=H$ with $q^{mon}=1/H$ (assuming ties are broken in the worst way, which could be ensured with a small perturbation) for
which the social welfare is only 1.  The 
serial monopoly price is $p^{ser}=1$ identical to the equilibrium price and so obtains full welfare of $1 + \ln H$.  (Here we had $p^{ser}=p^{eq}$ which can happen since in this example we have $Rev^{eq}=Rev^{mon}$.  With a tiny
perturbation we could have $Rev^{eq}<Rev^{mon}$ as assumed in our analysis of the dynamics which would then result in $p^{ser} > p^{eq}$ but arbitrarily close to it, still obtaining that the serial monopoly social welfare is arbitrarily close to the equilibrium welfare.)
\end{example}

\section{Strategic Equilibrium Analysis}\label{sec-eq}

In this section we attempt analyzing the dynamic behavior of this serial monopoly when the users are now strategic and ``patient''. 
I.e. while all of our previous analysis above assumed that each serial monopolist was faced with the true demand, we will
now assume that the monopolists are faced with the demand that is declared by the users.

\subsection{The Model}

We stick with our simple non-stochastic deterministic model with {\em price-taking bidders} and with the continuous and strictly decreasing daily demand distribution $Q$ and the supply $s$
being common knowledge.  We consider a game between
users and monopolists: for each time step $t$ we have a (newly born) continuum of infinitely-small users and a single (new) leader.  
Each of our users has a true value $v$, where, at each time step, $v$ is chosen according to the demand distribution $Q$.
Our game proceeds in time steps where at each time $t$ we have two stages.  In the first stage, each of the users ``born'' at this time 
declare a bid $\tilde{v}$ 
where each user with true value $v$ bids according to some strategic manipulation function $\tilde{v}=m_t(v)$.  We assume that users put 
their bid $\tilde{v}$ once when they first enter the market
rather than being able to change their bid every step.  Our users are fully patient and a user born at time $t$ with value $v$
gets utility $v-p^{t'}$ for the first $t' \ge t$ with $p^{t'} \le \tilde{v}$ (and $0$ if no such $t'$ exists).
The {\em declared daily demand} $\tilde{Q^t}()$ is generated by the distribution of these $\tilde{v}$'s and is added to the (declared) pent up demand
to obtain the total {\em declared} demand $\tilde{D}^t$ faced by the leader at time $t$.   In the second stage of this time step, the (newly born) leader of time $t$ gets to choose
a price $p^t$ for the current block and his reward for price $p^t$ is $p^t \cdot min(s,\tilde{D}^t(p^t))$.\footnote{Technically, our model allows the leaders to choose a price $p^t$ that leads to over-demand
$D(p^t) > s$, and this definition of utility of the users satisfies all the -- more than $s$ quantity -- of users who bid at least $p^t$ 
which may not be realistic.  
We could define precisely 
a quantity of exactly $s$ who get satisfied 
in such a case, but we do not have to worry about this here as our leaders will never (in any equilibrium) choose such $p^t$ since their own utility could be strictly increased by choosing $p'>p^t$
with $D(p')=s$ which gets higher revenue.} 

\subsection{An Equilibrium}

Let us start by identifying a natural equilibrium of this game:
First, the strategy of choosing $\tilde{p}^t$ to be the monopolist price for the declared 
total demand $\tilde{D}^t()$ at day $t$ is clearly a dominant strategy
for the leader of day $t$ so our equilibrium will fix this strategy for the leaders\footnote{Since the leaders get paid
according to the {\em declared} values, the leader's knowledge of the true demand does not allow him to do any better.}. 
We will have a fixed (time-independent) manipulation strategy for users: $\tilde{v} = m(v)$.
Once this user time-independent strategy function $m$ is chosen, the true demand distribution $Q()$ from which $v$ is chosen 
defines a declared demand distribution $\tilde{Q}()$
induced by $m(v)$.  The dynamics of the game then proceed as in our analysis in the previous 
sections but according to the declared distribution $\tilde{Q}()$ rather than the true distribution $Q()$.

\begin{example}
The user's strategy functions $m(v)=min(v,p^{eq})$ are in equilibrium with the leader's monopolist pricing for each day, where $p^{eq}$ is the equilibrium price of the true distribution $Q()$.  
In this equilibrium the daily prices and quantities do not fluctuate: $p^t=p^{eq}$ and $q^t=s$ for all $t$.
\end{example}

To show that this an equilibrium we first show that $p^t=p^{eq}$ is indeed the monopoly price for every leader.
Since with the declared demand induced by this $m$, the new daily declared demand at price $p^{eq}$ exactly exhausts 
the daily supply of the monopolist, $\tilde{Q}(p^{eq}=s$, and all pent-up demand is at lower values, 
maximizing revenue according to $\tilde{D}^t$
is the same as maximizing revenue according to $\tilde{Q}$ which is $p^{eq}$ for which
the declared demand is clearly $s$.  
I.e. in this case $p^{eq}$ is also the monopolist price of the declared $\tilde{Q}$ and thus there 
are no daily fluctuations in price but rather we always have $p^t=p^{eq}$.  
Now notice that the users' strategy function $m(v)=min(v,p^{eq})$ is indeed a best response since
users with $v < p^{eq}$ are not being served and they can only be served by paying $p^{eq}$ which they they are not willing to, 
while users with $v \ge p^{eq}$ are being served at price $p^{eq}$ and that is optimal for them.

\subsection{Characterization}

Technically, the equilibrium above is not the only Nash equilibrium point as, for example, we can have an equilibrium where everyone ``pretends'' 
that the supply is smaller than it really is $s' < s$, where all leaders only supply 
$s'$ and all bidders bid according to $m'(v)=p'^{eq}$ if $v \ge p'^{eq}$ and $m'(v)=0$ if $v < p'^{eq}$, where $p'^{eq}$ 
is the equilibrium price of the daily demand $Q()$ with the smaller pretended supply $s'$.
Note that with the given declared supply there is never more than $s'$ demand at a strictly positive 
price and so the leaders are best responding, while with these leader strategies there is a ``declared supply''
of only $s'$ and thus the users are best-replying.  This equilibrium is certainly artificial since
leaders are not using their dominant strategies.  

We will thus restrict ourselves to characterizing equilibria in {\em un-dominated strategies},
in which leaders will always use their dominant strategy of choosing the monopolist price 
(of the declared total demand at their day).  
Furthermore we will concentrate on manipulation strategies that are {\em time-invariant}, 
i.e. where for all $t$ the daily manipulation functions are the same $m_t(v)=m(v)$ for some fixed $m(v)$ (as in the example above)\footnote{This
time-invariance allows us to easily employ our analysis of the myopic case in section \ref{sec-dyn}, but our characterization does not {\em seem} to really require it.}
The strategy of users born at time $t$ will be given by the mapping $\tilde{v}=m(v)$
thus the induced declared new demand distribution at time $t$ is a fixed $\tilde{Q}$.  
I.e.  be an equilibrium if for every bidder that is born in time $t$ with value $v$, $m(v)$ is indeed a best reply
to all the other users bidding according to $m$ and all the leaders choosing the monopoly price.  So we have:

\begin{definition}
An equilibrium in un-dominated strategies with time-invariant manipulation is a manipulation function $m()$ such that for every user
born at time $t$ with value $v$, the bid $m(v)$ is a best reply to all the other users bidding according to $m()$ and all the leaders using
monopolist pricing.
\end{definition}

We now start analyzing the properties of such equilibria.

\begin{lemma} \label{no-fluc}
In every equilibrium in un-dominated strategies with time-invariant manipulation $m()$, we 
have that that for all $t$: $p^t=\tilde{p}^{eq}=\tilde{p}^{mon}$ where $\tilde{p}^{eq}$ and $\tilde{p}^{mon}$ are defined
with respect to the distribution $\tilde{Q}$ induced by $m$.
\end{lemma}

\begin{proof}
Un-dominated strategies (monopoly pricing) for the leaders and time-invariant manipulations 
($m()$) for the users lead us to exactly the myopic dynamics studied previously for the distribution $\tilde{Q}$ induced by $m()$.
Let $(\tilde{p}^{ser},\tilde{q}^{ser})$ be the serial monopoly price and quantity of the declared distribution.  
Our analysis in section \ref{sec-dyn} shows the dynamics with the declared demand
will essentially sell the quantity $\tilde{q}^{ser}$ (per day) at prices that fluctuate between $\tilde{p}^{mon}$ and $\tilde{p}^{ser}$.  At equilibrium, price fluctuations cannot occur since a bidder will never be willing to
pay a higher price if he can later get a lower price.  But as shown, such fluctuations will occur unless the equilibrium price $\tilde{p}^{eq}$ 
already gives the
monopolist's revenue.  It follows that at equilibrium $m()$ must induce a demand distribution $\tilde{Q}$ such 
that $\tilde{p}^{eq}=\tilde{p}^{mon}$ in which case we will have a fixed sale price $p^t=\tilde{p}^{eq}$ for all $t$.
\end{proof}

\begin{lemma}
In every equilibrium in un-dominated strategies with time-invariant manipulation $m()$
we have that for all $v<\tilde{p}^{eq}$ we have $m(v) < \tilde{p}^{eq}$ and for all $v>\tilde{p}^{eq}$ we have $m(v) \ge \tilde{p}^{eq}$. 
Hence $\tilde{Q}(\tilde{p}^{eq})=Q(\tilde{p}^{eq})$.
\end{lemma}

\begin{proof}
A bidder with $v<\tilde{p}^{eq}$ would rather lose than pay $p^t=\tilde{p}^{eq}$ so, as his bid does not affect 
the price $p^t$, he will have to bid lower than $\tilde{p}^{eq}$ for that to happen.  
The same (but opposite) is true for $v>\tilde{p}^{eq}$.
\end{proof}

We are now ready to prove theorem \ref{thm-eq}.

\begin{proof} (of theorem \ref{thm-eq})
The first part of the theorem was shown in the example in the previous subsection.  We now prove the second part.

As shown in lemma \ref{no-fluc}, in any such equilibrium we have a fixed sale price $p^t=\tilde{p}^{eq}$ for all $t$ and thus it 
suffices to show that $\tilde{p}^{eq}=p^{eq}$.  Since $\tilde{p}^{eq}$ is the equilibrium price of $\tilde{Q}()$ then 
$\tilde{Q}(\tilde{p}^{eq}) = s$ and
since, by the previous lemma, $\tilde{Q}(\tilde{p}^{eq})=Q(\tilde{p}^{eq})$ we also have $Q(\tilde{p}^{eq}) = s$.
But since $p^{eq}$ is the market equilibrium price of $Q$ we also have $Q(p^{eq})=s$
and since $Q$ was assumed to be strictly decreasing we must have thus $p^{eq} = \tilde{p}^{eq}$. 
\end{proof}

\section*{Acknowledgments}
I would like to thank
Abdelhamid Bakhta, Eli Ben-Sasson, Tom Brand, Ittay Dror, Lior Goldberg, Louis Guthmann, 
Oren Katz, Yoav Kolumbus, Avihu Levi, Tim Roughgarden, Ilya Voloch, and Aviv Zohar for useful discussions and comments.

\section*{Appendix: Postponed proofs of lemmas}

\begin{proof} (of lemma \ref{lem-a})
We can prove (1) and (2) by induction on $t$:
when moving from step $t-1$ to $t$ we have that
$D^{t}(p)-D^{t}(p')=(Z^{t-1}(p)-Z^{t-1}(p')) + (Q(p)-Q(p')) \le (D^{t-1}(p)-D^{t-1}(p')) + (Q(p)-Q(p'))$ and thus the LHS increases by at most
$((Q(p)-Q(p'))$ which is exactly how the RHS increases.  The base of the induction
holds as for part A, $t=1$, we have $Z^0(p)=Z^0(p')=0$, while for part B, $t=T+1$, we have $D^t(p)-D^t(p')=(Q(p) - Q(p')) + (Z^{T}(p)-Z^{T}(p'))$.

For (3), note that when $p^{t-1} \ge p' >p$ we have that $Z^{t-1}(p)=D^{t-1}(p)-q^{t-1}$ and $Z^{t-1}(p')=D^{t-1}(p')-q^{t-1}$ so
$D^t(p)-D^t(p')=(D^{t-1}(p)-D^{t-1}(p'))+(Q(p)-Q(p'))$ and thus we get equalities throughout the induction.

For (4) and (5), just note that $p^T \le p < p'$ (or $T=0$) we actually have $Z^{T}(p)=Z^{T}(p')=0$ and then apply (2) and (3) respectively.
\end{proof}

\begin{proof} (of lemma \ref{A1})
The maximum revenue that is achievable from a price $p$ is $p \cdot s$ and when $p < p^{ser}$ we have that 
$p \cdot s < p^{mon} \cdot q^{mon}$ and that revenue can be achieved at any step using the monopolist price.
\end{proof}

\begin{proof} (of lemma \ref{lem-b1})
For $p \ge p^{t-1}$ we have that $D^t(p)=Q(p)$ so the maximal revenue obtained by possible $p \ge p^{t-1}$ is exactly the monopolist's revenue that is obtained at $p^t=p^{mon}$ (we assume that ties in maximum 
revenue are broken
consistently).  So, unless $p^t=p^{mon}$ then we must obtain the maximum 
revenue in the range $p < p^{t-1}$.  
\end{proof}

\begin{proof} (of lemma \ref{lem-b2})
We will prove the first part of the lemma.  
First we cannot have $p^t=p^{mon}$ as the revenue obtained from $p$ would be higher: 
$p \cdot s > p^{ser} \cdot s = p^{mon} \cdot q^{mon}$.

As $p^t$ gives better revenue than $p$ i.e., we have that 
$p^t \cdot D^t(p^t) \ge p \cdot s = (p-p^{ser}) \cdot s + p^{ser} \cdot s = (p-p^{ser}) \cdot s + p^{mon} \cdot q^{mon}$.

Separating the total demand at time $t$ to its two components we get
$p^{t} \cdot D^{t}(p^{t}) = p^{t} \cdot Z^{t-1}(p^{t}) + p^{t} \cdot Q(p^{t}) \le p^{t} \cdot Z^{t-1}(p^{t}) + p^{mon} \cdot q^{mon} \le p^{mon} \cdot Z^{t-1}(p^{t}) + p^{mon} \cdot q^{mon}$. 

Putting these together we get that 
$(p-p^{ser}) \cdot s \le p^{mon} \cdot Z^{t-1}(p^{t})$.
Now $Z^{t-1}(p^{t}) = Z^{t-1}(p^t)-Z^{t-1}(p^{t-1}) \le D^{t-1}(p^t)-D^{t-1}(p^{t-1}) \le (t-1) \cdot (Q(p^t)-Q(p^{t-1})$ so it follows that
$(p-p^{ser}) \cdot s \le (t-1) \cdot p^{mon} \cdot (Q(p^t)-Q(p^{t-1})$ and thus $(Q(p^t)-Q(p^{t-1}) \ge (t-1)^{-1} \cdot s \cdot (p-p^{ser})/p^{mon}$.

The second part of the lemma is similar after taking into account that 
$D^{t-1}(p^t)-D^{t-1}(p^{t-1})$ is actually bounded by
$(t-1-T) \cdot (Q(p^t)-Q(p^{t-1})$.
\end{proof}

\begin{proof} (of lemma \ref{lem-c})
If we have that $p^t \ge p'$ for all $T \le t \le T+\Delta$, then using
lemma \ref{lem-a}(3) we have 
$D^{T+\Delta}(p) \ge D^{T}(p)-D^{T}(p') = \Delta \cdot (Q(p)-Q(p')) + (Z^T(p)-Z^T(p')) \ge \Delta \cdot (Q(p)-Q(p'))$.  
So just choose $\Delta_0 = s/(Q(p)-Q(p'))$.
\end{proof}

\begin{proof} (of lemma \ref{lem-main})
Assume not, ant let $T$ be some time step at which $p^t \le p^*$ or $T=0$ and let $p=(p^*+p^{ser})/2$ so 
$p^{ser} < p < p^*$ and $p^*-p^{ser} = 2 \cdot (p-p^{ser})$.  
By lemma \ref{lem-c} there exists $\Delta_0$ after which $D^{t}(p) \ge s$ for all 
$t \ge T+\Delta_0$ until the first time that $p^t \le p^*$.  
Fix any $\Delta > \Delta_0$ so that $p^t > p^*$ for all $T+\Delta_0 < t \le T+\Delta$. 
using lemma
\ref{lem-b2} we get a decreasing sequence of prices $p^{T+\Delta_0} > p^{T+\Delta_0+1} > p^{T+\Delta_0+2} > \cdots p^{T+\Delta}$
with $(Q(p^{t+1})-Q(p^{t})) \ge (t-T)^{-1} \cdot s \cdot (p-p^{ser}) / p^{mon}$.  
Summing up over all $T+\Delta_0 < t \le T+\Delta$ we get 
$Q(p^{T+\Delta})-Q(p^{T+\Delta_0}) \ge (\sum_{t=T+\Delta_0+1}^{T+\Delta} (t-T)^{-1}) \cdot s \cdot (p-p^{ser}) / p^{mon}$.  
We now estimate
$\sum_{t=T+\Delta_0+1}^{T+\Delta} (t-T)^{-1} = \sum_{i=\Delta_0+1}^{\Delta} i^{-1} \ge \ln(\Delta/(1+\Delta_0))$.
So $Q(p^{T+\Delta})-Q(p^{T+\Delta_0}) \ge \ln(\Delta/(1+\Delta_0)) \cdot s \cdot (p-p^{ser}) / p^{mon}$.
Since $Q(p^{T+\Delta})-Q(p^{T+\Delta_0}) \le Q(p^{ser})-Q(p^{mon})$, whenever $Q(p^{ser})-Q(p^{mon}) < \ln(\Delta/(1+\Delta_0)) \cdot s \cdot (p-p^{ser}) / p^{mon}$ 
then we to get a contradiction.  
I.e. if we choose $\Delta$ so that $\ln(\Delta)  > \ln(1+\Delta_0) + 2 \cdot (Q(p^{ser})-Q(p^{mon})) \cdot p^{mon} / (s \cdot (p^*-p^{ser}))$
then at some step $T+\Delta_0 \le t \le T+\Delta$ we must have $p^t \le p^*$.  
\end{proof}

\begin{proof} (of lemma \ref{lem-q})
The proof is by induction on $T$. 
For $T=1$, the LHS is $0$ and the RHS is $0$.
When moving from step $T-1$ to $T$, the LHS grows by exactly $q^{T-1}$.  
The first term on the RHS grows
by $Q(p)$ and $Z^T(p) \ge D^{T-1}(p)-q^{T-1}$ and thus the second term on the RHS,
$D^T(p)=Z^T(p)+Q(p)$, grows by
at least $Q(p)-q^{T-1}$, as needed.  
When $p \le p^{T-1}$ then we have that $Z^T(p) = D^{T-1}(p)-q^{T-1}$ and the second term
grows by exactly the required amount.
\end{proof}

\begin{proof} (of lemma \ref{lem-up})
Assume by way of contradiction that that there is some last time where $p^{t} = p^{mon}$ and thus by lemma \ref{lem-b1} after this time the prices $p^t$ are a monotone decreasing sequence and so by lemma \ref{lem-main}
they approach $p^{ser}$.  Let $q^{ser} < q^* < s$ (see section \ref{sec-seq}), let $p^*>p^{ser}$ be 
so that $p^* \cdot q^* <  p^{ser} \cdot s = p^{mon} \cdot q^{mon}$ (such a 
value for $p^*$ must exist since $Q()$ is continuous), and let $T_0$ be a point
for which for every $t > T_0$ we have $p^t < p^*$ (which must exist according to lemma \ref{lem-main}).  

Since $p^t$ optimizes revenue at time $t$ we also must have $p^t \cdot q^t \ge p^* \cdot q^*$ and since 
$p^t < p^*$ for $t > T_0$ we must have $q^t > q^*$ for all $t > T_0$.  It follows that the total supplied quantity up to some large
time $T>T_0$ is $\sum_{t=1}^{T} q^t \ge (T-T_0) \cdot q^*$.  We now apply lemma \ref{lem-q} to get 
$\sum_{t=1}^{T} q^t = (T+1) \cdot Q(p^{ser}) - D^T(p^{ser}) \le (T+1) \cdot q^{ser}$.  
Putting these together we have that $(T-T_0) \cdot q^* \le (T+1) \cdot q^{ser}$ which is a contradiction for large enough $T$ since
$q^* > q^{ser}$.
\end{proof}

\end{document}